\documentclass[12pt,notitlepage]{article}%
\usepackage{tikz}
\usepackage{amsthm}
\usepackage{setspace}
\usepackage{eurosym}
\usepackage{subfigure}
\usepackage{xcolor}
\usepackage{amssymb}
\usepackage{graphicx}
\usepackage[colorlinks,linkcolor=links,citecolor=cites,urlcolor=MyDarkBlue]%
{hyperref}
\usepackage{amsfonts}
\usepackage{amsmath}
\usepackage{amstext}
\usepackage{appendix}
\usepackage{appendix}
\usepackage{mathpazo}
\usepackage{natbib}
\providecommand{\U}[1]{\protect\rule{.1in}{.1in}}
\usetikzlibrary{positioning,fit,shapes,arrows.meta}
\newtheorem{theorem}{Theorem}

\newtheorem{corollary}{Corollary}

\newtheorem{definition}{Definition}
\newtheorem{example}{Example}

\newtheorem{lemma}{Lemma}

\newtheorem{proposition}{Proposition}
\newtheorem{remark}{Remark}

\newtheorem*{discussion*}{Discussion}
\numberwithin{equation}{section}
\definecolor{MyDarkBlue}{rgb}{0,0.08,0.45}
\definecolor{cites}{HTML}{324b13}
\definecolor{links}{HTML}{1a663b}
\definecolor{MyLightMagenta}{cmyk}{0.1,0.8,0,0.1}
\hypersetup{
	colorlinks,citecolor=blue,filecolor=black,linkcolor=blue,urlcolor=blue
}

\graphicspath{{D:/OneDrive/TEX/PicsforPapers/}}
\begin{document}

\title{Credible Nash Bargaining Solution for Bilateral Trading Networks \thanks{We thank Gaoji Hu, Herve Moulin, Jie Zheng, and participants at the 2024 Shanghai Microeconomics Workshop, 2024 HEAT Workshop at HUST, and 2026 iES Mechanism Design Workshop at SWUFE for helpful comments and discussions. This project is supported by NSFC (No. 72273085 and No. 72033004).}}
\author{Kang Rong\thanks{School of Economics, Shanghai University of Finance and Economics. \\
\hspace*{1.8em}	\textit{Email addresses:} rong.kang@mail.shufe.edu.cn (Rong), tang.qianfeng@mail.shufe.edu.cn (Tang).}
	\hspace{0.6cm}	Qianfeng Tang\footnotemark[2]}
\maketitle

\begin{abstract}
	
We study surplus division in network-constrained bilateral matching markets with transferable utility. We introduce a new solution concept—the credible bargaining solution—which refines stability by requiring that, for each matched pair of buyer and seller, surplus be divided according to the Nash bargaining solution with respect to credible outside options, defined as their payoffs in some stable outcome of the submarket obtained by removing their link. We establish general properties of the credible bargaining solution, prove existence, and provide a complete characterization in the unit-surplus case based on the notion of essential links.

\bigskip


\end{abstract}

\newpage

\section{Introduction}

We consider a market consisting of a finite set of buyers and a finite set of sellers. Each seller is endowed with a single indivisible object, and each buyer demands exactly one unit. Trading opportunities are constrained by a network: a buyer and a seller can trade only if they are connected in the network. The network structure may capture informational, institutional, or physical constraints on exchange. Beyond standard exchange markets, this framework encompasses a range of applications, including marriage markets with transfers, labor markets, and other bilateral matching environments.

In short, the market we study is an assignment game of \cite{ShapleyShubik72}, restricted to a trading network. An outcome of the market specifies a matching between buyers and sellers together with a payoff for each agent. A natural requirement for an outcome is stability: there should be no buyer–seller pair that is connected in the network, remains unmatched, and can jointly benefit by matching with each other and appropriately sharing the surplus from trade. Stability imposes constraints across matched pairs and ensures the absence of blocking deviations. However, it places only weak restrictions on the division of surplus within each matched pair. In particular, a wide range of payoff distributions, including highly unequal ones, may satisfy stability.\footnote{See \cite{ShapleyShubik72} and \cite{Rochford84} for details.}

This observation motivates our focus on how a matched buyer–seller pair bargains over the surplus generated by trade. When bargaining between a given pair breaks down, each party may turn to alternative partners available through the network, thereby realizing an outside option. These outside options are endogenously determined by the structure of the network and play a critical role in shaping bargaining outcomes. Understanding how outside options arise and how they affect surplus division is therefore essential for analyzing networked matching markets.\footnote{See \cite{Nash53}, \cite{Harsanyi56}, \cite{Harsanyi77}, and \cite{Crawford-Rochford86} for discussions on variable and endogenous outside options in bargaining.}

We model bilateral bargaining as follows. When a buyer $i$ and a seller $j$ are matched and begin to bargain, their outside options are determined by the submarket obtained by removing the link between $i$ and $j$. Intuitively, if bargaining between $i$ and $j$ breaks down, each agent may seek alternative trading partners in the network, but the two can no longer trade with each other. We say that the outside options $(d_i, d_j)$ of buyer $i$ and seller $j$ are credible if they coincide with the respective payoffs of $i$ and $j$ in some stable outcome of this submarket. Imposing credibility therefore amounts to assuming that the two agents agree on which outcome will arise in the submarket and that this outcome is stable.

An outcome is said to be a credible bargaining solution of the market if it satisfies two conditions. First, the outcome is stable. Second, for every matched buyer–seller pair, the division of surplus is determined by the Nash bargaining solution with respect to some pair of credible outside options. That is, each agent receives her outside option, and the remaining surplus is split equally between the buyer and the seller. In this definition, the second condition imposes a justifiability requirement on surplus division within each matched pair, while the first condition imposes no-blocking-deviation constraints across matched pairs.

Our analysis begins by establishing several general results. First, we show that if a pair of agents is matched in some optimal matching—that is, a matching that maximizes total surplus—then for any pair of credible outside options $(d_i, d_j)$ they may propose, the sum of these outside options does not exceed the surplus generated by matching the two agents. As a consequence, the bargaining problem between them is always well defined.

Second, we show that if a vector of payoffs is part of some credible bargaining solution, then it remains part of a credible bargaining solution when combined with any optimal matching. This result implies that, like stability, the credible bargaining solution primarily restricts the distribution of payoffs rather than the specific matching of agents.

Third, we show that at any stable outcome, if a matched pair does not constitute an essential link of the market, then the division of surplus within that pair is always justifiable as a Nash bargaining solution under some pair of credible outside options. A link is said to be essential if it is matched in every optimal matching. This result implies that, in order to verify whether a stable outcome qualifies as a credible bargaining solution, it suffices to examine surplus division only on matched essential links.

Our main results consist of an existence result and a characterization result for the unit-surplus case. To establish existence, we show that the average of the buyer-optimal and buyer-minimal stable payoff vectors, when paired with any optimal matching, always constitutes a credible bargaining solution. We refer to this outcome as the solution with symmetric compromise. To justify the surplus division within each matched pair $i$ and $j$, each agent claims as her outside option the optimal stable payoff she would obtain in the submarket formed by removing the link $ij$, reduced by a common compromise amount $\delta$. We show that such a compromise always exists to ensure that the resulting claims constitute credible outside options.

Our characterization result shows that, in markets in which every matched pair generates exactly one unit of surplus, an outcome is a credible bargaining solution if and only if it is stable and, for every matched essential link  $ij$, both $i$ and $j$ receive $1/2$. By our earlier results, it suffices to verify surplus division only on matched essential links. The characterization relies on the key observation that if a matched link $ij$ is essential, then once this link is removed from the network, there exists, for each of $i$ and $j$, an optimal matching in which that agent remains unmatched. Consequently, the stable payoff—and hence the outside option—of each agent in the submarket is zero. We further provide a graph-theoretic procedure to identify essential links. Specifically, we first apply the Edmonds–Gallai decomposition to the network and then further decompose the perfectly matched component into elementary subgraphs. A link is essential if and only if it forms an elementary subgraph by itself.

Our study is most closely related to the classical contributions of \cite{ShapleyShubik72}, \cite{Rochford84}, \cite{CookYamagishi92}, and \cite{Crawford-Rochford86}. The market we study is essentially an assignment game in the sense of \cite{ShapleyShubik72}, that is, a matching problem with transferable utility. Many foundational results on assignment games—including the existence and structure of optimal matchings, the core, and competitive equilibria—originate from \cite{ShapleyShubik72}.

\cite{Rochford84} observes that the set of stable outcomes in assignment games is typically large and may contain outcomes with highly unequal payoffs. To refine this indeterminacy, she proposes the symmetrically pairwise-bargained (SPB) solution, which selects more egalitarian outcomes from the core. An equivalent concept, known as the balanced outcome, is later introduced by \cite{CookYamagishi92} in the sociology literature. Both concepts assume that, when a matched pair of agents bargains, their outside options are given by their maximal deviation payoffs, taking the payoffs of all other agents as fixed. In particular, a breakdown in bargaining between a given pair does not affect the payoffs or expectations of other agents. This assumption is fundamentally different from ours: in our framework, outside options are determined by the stable outcomes of submarkets that arise when links are removed. We discuss the SPB solution and the balanced outcome in detail in Subsection \ref{sec:balanced_outcome}.

To our knowledge, \cite{Crawford-Rochford86} is the first to define outside options by explicitly considering the submarket in which the surplus from matching a given pair of agents is eliminated. They define a notion of bargaining equilibrium recursively: within every matched pair, outside options must coincide with payoffs from some bargaining equilibrium of the corresponding submarket. Moreover, their concept requires that no agent has an incentive to deviate to bargain with any other agent. In contrast, our definition of a credible bargaining solution is not recursive, as outside options are taken from stable outcomes of the submarket. Furthermore, the no-deviation requirement in \cite{Crawford-Rochford86} differs conceptually from the no-deviation requirement imposed by stability.

The credible bargaining solution we propose is a cooperative concept; it abstracts from the details of bargaining protocols and strategic complexity that are central to noncooperative models. Important works that study noncooperative bargaining on networks include \cite{Corominas04}, \cite{Polanski07}, \cite{AbreuManea12GEB}, and \cite{AbreuManea12JET}.\footnote{All of these papers focus on the study of the network structure, by assuming that all links generate exactly one unit of surplus.} Among them, \cite{Corominas04} and \cite{Polanski07} study dynamic bargaining processes that are relatively centralized, in the sense that the selection of simultaneous trades or offer-making opportunities is governed by maximum matchings. In contrast, \cite{AbreuManea12GEB} and \cite{AbreuManea12JET} adopt a more decentralized approach, assuming that in each period a single linked pair is randomly selected to bargain. \cite{Manea16} provides a comprehensive comparison of these models and their equilibrium predictions.

The rest of this paper is organized as follows. In Section \ref{sec:prelim}, we introduce the notations and the classical concepts of stable outcomes and balanced outcomes. In Section \ref{sec:CBS}, we propose the credible bargaining solution, and in Section \ref{sec:results}, we present our analysis on this concept. We conclude in Section \ref{sec:concl}.

\section{Preliminaries}\label{sec:prelim}

\subsection{Notations}

There are a finite set of buyers $I$ and a finite set of sellers $J$ in a bilateral trading network which we call the market. These buyers and sellers are also called agents. Each seller has exactly one unit of object to sell and each buyer demands exactly one unit of object. If buyer $i$ buys from seller $j,$ this will generate a surplus $v_{ij}\geq0$. In the market, buyers and sellers are connected by an undirected bipartite network (graph) $G,$ which we define as a subset of $I\times J$. If $(i,j) \in G,$ then $(i,j)$ is called a link of the network $G$ and buyer $i$ and seller $j$ are said to be connected by the network $G$ via the link $(i,j)$. When no confusion may arise, we simply use $ij$ to denote a link $(i,j)$, and we use $G_{-ij}$ to denote the subgraph $G\backslash\{(i,j)\}$. Likewise, we use $G_{-i}\equiv\left\{  (i^{\prime},j^{\prime}) \in G:i^{\prime}\neq i \right\}$ to denote the subgraph of $G$ in which all of agent $i$'s links are removed, and we use $G_{-i,-j}\equiv\left\{(i^{\prime},j^{\prime})\in
G:i^{\prime}\neq i,j^{\prime}\neq j\right\}  $ to denote the subgraph of $G$ in which all of agent $i$'s links and agent $j$'s links are removed. 

For convenience, we denote the market by $(G,v),$ where $v\equiv (v_{ij})_{i\in I,j\in J}$. In the market, a buyer $i$ may buy from a seller $j$ only if they are connected, i.e., only if $ij\in G$. For $ij\in G,$ we may also call $v_{ij}$ as link $ij$'s valuation. When every buyer is connected with every seller, the market represents an assignment game of \cite{ShapleyShubik72}, with valuation matrix $v$.

A matching is a one-to-one function $\mu:I\cup J \rightarrow I\cup J\cup \{\emptyset\},$ such that (i) $\mu(i)  \in J \cup \{\emptyset\}$ for all $i\in I,$ and $\mu(j)\in I\cup \{\emptyset\}  $ for all $j\in J;$ (ii) $\mu (i)  =j$ iff $\mu(j)  =i,$ for all $i\in I,j\in J$; and (iii) $\mu(i)  =j\in J$ only if $ij\in G$. For any buyer $i$ and seller $j,$ if $\mu(i)  =j,$ we say that they are matched by $\mu$. If for $k\in I\cup J,\mu(k)  =\emptyset,$ we say that agent $k$ is unmatched.

Every matching $\mu$ can also be represented by the set of disjoint links in $G$ that it matches. If $\mu(i)=j$, we also write $ij \in\mu$.

An outcome of the market $(G,v)$ is a pair $(\mu,x)$, where $\mu$ is matching and $x\in\mathbb{R}^{I}\times\mathbb{R}^{J}$ is a payoff vector. Throughout, we consider only outcomes that are feasible. An outcome $(\mu, x)$ is feasible if $x$ is a division of the total surplus generated by $\mu$, i.e., if
\[
\sum_{i\in I}x_{i}+ \sum_{j\in J}x_{j}=\sum_{ij\in\mu}v_{ij}.
\]

\subsection{Stable outcome}

\begin{definition}
An outcome $\left(  \mu,x\right)  $ is \textbf{stable} if for all $i\in I,j\in J:$

\begin{enumerate}
\item $x_{i},x_{j}\geq0;$

\item $x_{i}+x_{j}\geq v_{ij}, \forall ij \in G.$
\end{enumerate}
\end{definition}

If $(\mu,x)$ is a stable outcome, then $x_{i}+x_{j}=v_{ij}$ for any matched pair $ij$, and $x_{k}=0$ for all unmatched agent $k$.

A matching $\mu$ is said to be \textbf{optimal} if the total surplus it generates, $\sum_{ij\in\mu}v_{ij},$ is the highest among all matchings. It is well-known that the matching in any stable outcome is an optimal matching, and the set of payoff vectors in stable outcomes coincides with the core and the set of competitive equilibrium payoff vectors. Also, if $\mu$ is any optimal matching and $x$ is the payoff vector of any stable outcome, then $(\mu,x)$ is also a stable outcome. Therefore, if an agent is unmatched by any optimal matching, then her payoff is zero across all stable outcomes. 

For details of these results, see \cite{ShapleyShubik72} or Chapter 8.1 of \cite{RothSoto90}.\footnote{See also \cite{NunezRafels15} for a more recent survey on assignment markets.}

\begin{example}\label{ex:stability}
Suppose $I=\{1, 3\}$, $J=\{2, 4\}$, $G=\{12, 23, 34\}$, and $v_{ij} =1$ for
all links of $G$.

\begin{figure}[h]
\centering
\begin{tikzpicture}[
		every node/.style={circle, draw, minimum size=1cm},
		every edge/.style={draw, thick, ->}
		]
		
		% Nodes in a line
		\node (1) at (0,0) [label={[above=-0.1cm]: $x_1$}] {1};
		\node (2) [right=of 1] [label={[above=-0.1cm]: $x_2$}] {2};
		\node (3) [right=of 2] [label={[above=-0.1cm]: $x_3$}] {3};
		\node (4) [right=of 3] [label={[above=-0.1cm]: $x_4$}] {4};
		
		% Edges
		\draw (1) -- (2);
		\draw (2) -- (3);
		\draw (3) -- (4);
		
	\end{tikzpicture}
	\caption{Four agents connected by a line}
	\label{fig:four-line}
	
\end{figure}
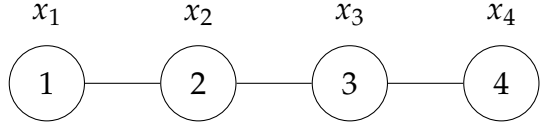

Then $(\mu, x)$ is a stable outcome iff $\mu=\{12, 34\}$, and $x$ satisfies
\[
x_{1} +x_{2}=x_{3}+x_{4} =1, x_{2} +x_{3} \geq1.
\]

In particular, both $(0, 1, 0, 1)$ and $(1, 0, 1, 0)$ are payoff vectors in stable outcomes. This implies that in stable outcomes, the division of payoffs between matched agents may be extremely unfair.

\end{example}

\subsection{Balanced outcome} \label{sec:balanced_outcome}

Before introducing our concept in the next section, we review the concept of balanced outcome proposed by \cite{Rochford84} and \cite{CookYamagishi92}, which selects, among stable outcomes, outcomes in which the divisions of surplus among matched agents are more balanced. Its definition extends the Nash bargaining solution to bilateral trading networks.

Suppose buyer $i$ and seller $j$ bargain over their surplus $v_{ij}$, with outside options $d_{i}$ and $d_{j}$, respectively, where $d_{i} + d_{j} \leq v_{ij}$. The Nash bargaining solution (\citealp{Nash50}) between them, denoted by $NBS(v_{ij}; d_{i}, d_{j})$, is a pair $(x_{i}, x_{j})$ such that \footnote{In our setting of linear utilities, the Nash bargaining solution coincides with the intuitive "split-the-difference" rule, which has been applied since ancient times; see, e.g., Aristotle's discussion on justice (\citealp{Aristotle04}, Book V, Ch. 4) and Aumann and Maschler's discussion on contested claims (\citealp{AumannMaschler85}).}

\begin{align*}
	x_{i}  &  =d_{i} + \frac{1}{2} (v_{ij}-d_{i} -d_{j});\\
	x_{j}  &  =d_{j} + \frac{1}{2} (v_{ij}-d_{i} -d_{j}).
\end{align*}

A balanced outcome is a stable outcome $(\mu,x)$ such that for each matched pair $ij \in \mu$, their payoffs $x_{i}$ and $x_{j}$ are the division of $v_{ij}$ according to the Nash bargaining solution, with the outside option $d_k$ of each agent $k \in \{i, j\}$ being her largest benefit from deviation.

\begin{definition} [\citealp{Rochford84}; \citealp{CookYamagishi92}]
	A stable outcome $(\mu, x)$ is a \textbf{balanced outcome} if for every $ij \in \mu$, $(x_i, x_j)=NBS(v_{ij}; d_i, d_j)$, where for each $k\in \{i, j\}$, 
	\[d_k=\max\left\{0, \max_{kl \in G_{-ij}} (v_{kl}-x_l)\right\}.\]
\end{definition}

By definition, as long as $(\mu, x)$ is stable, the sum of $i$ and $j$'s outside options, $d_i + d_j$ is no more than $v_{ij}$.

An equivalent concept, called symmetrically pairwise-bargained (SPB) allocation, was proposed earlier by \cite{Rochford84} for assignment games. \cite{Rochford84} characterizes the set of SPB allocations as the intersection of the kernel (\citealp{Davis-Maschler65}) and the core (i.e., the set of stable outcomes) and also as the set of fixed points of a rebargaining process. Later, \cite{Driessen98} shows that for assignment games, the kernel is actually a subset of the core. Therefore, the set of SPB allocations (i.e., balanced outcomes) coincides with the kernel.\footnote{\cite{RothSoto88} show that the set of SPB allocations is a lattice. See also \cite{KleinbergTardos08} for a computational characterization of balanced outcomes.}

\begin{example} \label{ex:balanced-outcome}
Let's revisit the market in Example \ref{ex:stability}, where four agents are linked by a line. In this market, an outcome $(\mu, x)$ is a balanced outcome if and only if $\mu=\{12, 34\}$ and $x=(1/3,2/3,2/3,1/3)$.

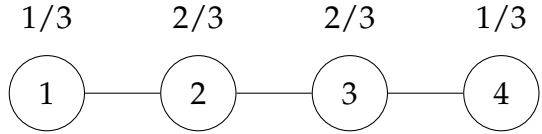
\begin{figure}[h]
\centering
\begin{tikzpicture}[
		every node/.style={circle, draw, minimum size=1cm},
		every edge/.style={draw, thick, ->}
		]
		
		% Nodes in a line
		\node (1) at (0,0) [label={[above=-0.1cm]: 1/3}] {1};
		\node (2) [right=of 1] [label={[above=-0.1cm]: 2/3}] {2};
		\node (3) [right=of 2] [label={[above=-0.1cm]: 2/3}] {3};
		\node (4) [right=of 3] [label={[above=-0.1cm]: 1/3}] {4};
		
		% Edges
		\draw (1) -- (2);
		\draw (2) -- (3);
		\draw (3) -- (4);
		
	\end{tikzpicture}
	\caption{Balanced outcome payoffs}
	\label{fig:balanced-outcome}
\end{figure}

When agents $1$ and $2$ bargain, by definition of the balanced outcome, buyer $1$'s outside option (i.e., her maximal benefit from deviation) is $0$, while seller $2$'s outside option is $1-x_3=\frac{1}{3}$. With these outside options, $NBS(v_{12}; d_1, d_2)=(x_1, x_2)$. Likewise, $ NBS(v_{34}; d_3, d_4)=(x_3, x_4)$.

\end{example}

In the example above, when seller $2$ threatens buyer $1$ by claiming that she will receive $1/3$ if she deviates to match with buyer $3,$ buyer $1$ should find this threat \textbf{non-credible}. This is because if seller $2$ indeed breaks up with buyer $1$, she (and buyer $1$) will face the submarket $(G_{-12}, v)$, in which the only stable outcome payoff for seller $2$ is $0$. 

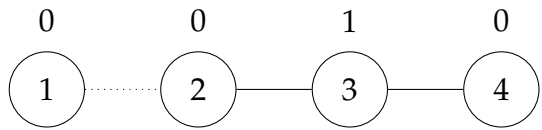
\begin{figure}[h]
	\centering
	\begin{tikzpicture}[
		every node/.style={circle, draw, minimum size=1cm},
		every edge/.style={draw, thick, ->}
		]
		
		% Nodes in a line
		\node (1) at (0,0) [label={[above=-0.1cm]: 0}] {1};
		\node (2) [right=of 1] [label={[above=-0.1cm]: 0}] {2};
		\node (3) [right=of 2] [label={[above=-0.1cm]: 1}] {3};
		\node (4) [right=of 3] [label={[above=-0.1cm]: 0}] {4};
		
		% Edges
		\draw[dotted] (1) -- (2);
		\draw (2) -- (3);
		\draw (3) -- (4);
		
	\end{tikzpicture}
	\caption{The graph $G_{-12}$ and stable outcome payoffs}
	\label{fig:G-12}
\end{figure}

\section{Credible bargaining solution}\label{sec:CBS}

\subsection{Definition}

Fix a market $(G, v)$. Suppose buyer $i$ and seller $j$ are matched and they decide to divide their surplus $v_{ij}$ according to the Nash bargaining solution. Once they agree on their respective outside options, $d_{i}$ and $d_{j}$, they will reach the division $NBS(v_{ij}; d_{i}, d_{j})$. 

However, given that $i$ and $j$ belong to a bilateral trading network, their outside options--which capture what they expect to receive if their bargaining breaks down--depend on their expectations on what is going to happen in the submarket $(G_{-ij}, v)$. We make the following credibility assumption on $i$ and $j$'s expectations on their outside options. 

\begin{definition}
	When a matched pair of agents $i$ and $j$ bargain to divide $v_{ij}$, their outside options $d_{i}$ and $d_{j}$ are called (jointly) \textbf{credible} if they are their respective payoffs in some stable outcome $(\mu^{\prime}, x^{\prime})$ of the submarket $(G_{-ij}, v)$.
\end{definition}

That is, $i$ and $j$ share the point of view that, if their bargaining fails, they will both face the submarket $(G_{-ij}, v)$ and this market will reach a stable outcome. Their outside options, therefore, should be their payoffs in that stable outcome. When there are multiple stable outcomes in $(G_{-ij}, v)$, $i$ and $j$ may perceive any of them to realize, but they need to agree on the same one.

\begin{definition}
An outcome $(\mu, x)$ is a \textbf{credible bargaining solution} of the market $(G, v)$ if

\begin{enumerate}
\item it is stable; and

\item for every $ij \in\mu$, their division $(x_i, x_j)$ is justifiable (by Nash bargaining under credible outside options): there exist credible outside options $d_{i}$ and
$d_{j}$ such that
\[
(x_{i}, x_{j})=NBS(v_{ij}; d_{i}, d_{j}).
\]

\end{enumerate}
\end{definition}

In this definition, the second condition requires that at $(\mu, x)$, the division $(x_{i}, x_{j})$ between any matched pair $i$ and $j$ is justifiable as the Nash bargaining solution under some credible outside options. Since the second condition is imposed separately on matched pairs, it does not ensure that the resulted division across matched pairs satisfy stability, i.e., agents may still have incentives to deviate. At an outcome $(\mu, x)$ that satisfies the second condition, if agent $j$ has a profitable deviation, then this will drive $i$ and $j$ to agree on a new pair of credible outside options $(d'_i, d'_j)$ that are more favorable to $j$, which in turn justifies a new division that increases $j$'s share from bargaining. It is exactly such dynamics that motivates our imposition of the stability requirement, i.e., the first condition, on $(\mu, x)$.

\subsection{Illustrative examples}

In this subsection, we use three simple networks to illustrate the credible bargaining solution. Later, in Example \ref{example_EG}, we present a larger network which contains elements of all these simple networks.

\begin{example}
Let's revisit the market we discussed in Examples \ref{ex:stability} and \ref{ex:balanced-outcome}, where four agents are connected by a line and every link generates one unit of surplus. An outcome $(\mu, x)$ is a credible bargaining solution if and only if $\mu=\{12, 34\}$ and $x=\{1/2, 1/2, 1/2, 1/2\}$. 
	
	\begin{figure}[h]
		\centering
		\begin{tikzpicture}[
			every node/.style={circle, draw, minimum size=1cm},
			every edge/.style={draw, thick, ->}
			]
			
			% Nodes in a line
			\node (1) at (0,0) [label={[above=-0.1cm]: 1/2}] {1};
			\node (2) [right=of 1] [label={[above=-0.1cm]: 1/2}] {2};
			\node (3) [right=of 2] [label={[above=-0.1cm]: 1/2}] {3};
			\node (4) [right=of 3] [label={[above=-0.1cm]: 1/2}] {4};
			
			% Edges
			\draw (1) -- (2);
			\draw (2) -- (3);
			\draw (3) -- (4);
			
		\end{tikzpicture}
		
		\caption{Credible bagaining solution payoffs}
		\label{fig:4-line-CBS}

	\end{figure}
	
To see this, consider the bargaining between buyer $1$ and seller $2$. Note that in $G_{-12}$, seller $2$'s only stable payoff is $0$. Therefore, $d_1=d_2=0$ and hence $NBS(v_{12}; d_1, d_2)=(1/2, 1/2)$. The bargaining between $3$ and $4$ can be verified in the same way. Lastly, this outcome is stable.

This market also highlights the difference between the credible bargaining solution and the balanced outcome; the latter predicts the payoff vector $(1/3, 2/3, 2/3, 1/3)$ (see Figure \ref{fig:balanced-outcome}). 
\end{example}

\begin{example} \label{3-line}
	
Suppose $I=\{1, 3\}, J=\{2\}, G=\{12, 23\}$, and $v_{ij}=1$ for all links of $G$. Then $(\mu, x)$ is a credible bargaining solution if and only if it is a stable outcome, which means $\mu=\{12\}$ or $\{23\}$, and $x=\{0, 1, 0\}$.

\begin{figure}[h]
	\centering
	\begin{tikzpicture}[
		every node/.style={circle, draw, minimum size=1cm},
		every edge/.style={draw, thick, ->}
		]
		
		% Nodes in a line
		\node (1) at (0,0) [label={[above=-0.1cm]: 0}] {1};
		\node (2) [right=of 1] [label={[above=-0.1cm]: 1}] {2};
		\node (3) [right=of 2] [label={[above=-0.1cm]: 0}] {3};
		
		% Edges
		\draw (1) -- (2);
		\draw (2) -- (3);
		
	\end{tikzpicture}
	\caption{Credible bargaining solution payoffs}
	\label{fig:3-line-CBS}
\end{figure}

	Let's see how the unique stable outcome payoff vector is part of a credible bargaining solution. Suppose $\mu=\{12\}$. The division $(x_1, x_2)=(0, 1)$ can be justified as the Nash bargaining solution between buyer $1$ and seller $2$, given outside options $(d_1, d_2)=(0, 1)$. This pair of outside options are part of the most favorable stable outcome for seller $2$ in the network $G_{-12}$.
	
	If buyer $1$ and seller $2$ instead agree on another pair of credible outside options, $(d'_1, d'_2)=(0, 0)$, then their division will be $(x'_1, x'_2)=(1/2, 1/2)$. Given that the unmatched agent--buyer $3$--is receiving $0$, the induced outcome is not stable: seller $2$ has incentive to rematch with buyer $3$. Such a deviation threat will drive buyer $1$ to accept credible outside options that are more favorable to seller $2$, and in the end, $d_2=1$. 
	
\end{example}

\begin{example}

    Suppose $I=\{1, 3\}, J=\{2, 4\}, G=\{12, 23, 34, 41\}$ and $v_{ij}=1$ for all links of $G$. Then $(\mu, x)$ is a credible bargaining solution if and only if it is a stable outcome, which means $\mu=\{12, 34\}$ or $\{14, 23\}$, and $x=\{\alpha, 1-\alpha, \alpha, 1-\alpha\}$ for some $\alpha \in [0, 1]$.
	
	\begin{figure}[h]
		\centering
		\begin{tikzpicture}[
			every node/.style={circle, draw, minimum size=1cm},
			every edge/.style={draw, thick, ->}
			]
			
			% Nodes
			\node (1) at (0,0) [label={[above=-0.1cm]: $\alpha$}] {1};
			\node (2) [right=of 1] [label={[above=-0.1cm]: $1-\alpha$}] {2};
			\node (3) [below=of 2] [label={[below=0.9cm]: $\alpha$}] {3};
			\node (4) [below=of 1] [label={[below=0.9cm]: $1-\alpha$}] {4};
			
			% Edges
			\draw (1) -- (2);
			\draw (2) -- (3);
			\draw (3) -- (4);
			\draw (4) -- (1);
			
		\end{tikzpicture}
		
		\caption{Four agents connected by a square}
		\label{fig:square}
	\end{figure}
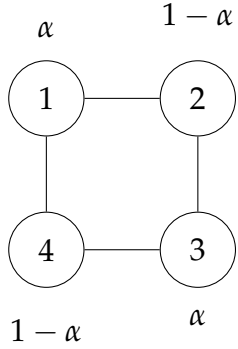
	
	For this market, the set of balanced outcomes and the set of credible bargaining solutions coincide.
\end{example}

\subsection{Discussion on recursivity}

Unlike recursive bargaining solutions, which define outside options through the solution itself applied to submarkets, our concept takes stable outcomes of submarkets as primitives. This avoids recursion and yields a more transparent and tractable notion of credibility.

A solution concept is intended to predict outcomes that are likely to arise in a market. In defining the credible bargaining solution, we impose only mild and standard assumptions on agents’ behavior. Specifically, we assume that bilateral bargaining follows the Nash bargaining solution, and that when bargaining between matched agents $i$ and $j$ breaks down, they expect the submarket $(G_{-ij}, v)$ to resolve at a stable outcome. Importantly, unlike the recursively defined bargaining equilibrium of \cite{Crawford-Rochford86}, our approach does not require agents to anticipate or agree upon the solution concept itself—either in the original market or in the relevant submarkets.

From a practical perspective, imposing recursion on our concept would also lead to severe existence problems. To illustrate, consider the market in Example \ref{3-line}. For an outcome $(\mu, x)$ to be stable, it must be that $\mu=\{12\}$ or $\{23\}$, and $x=(0, 1, 0)$. Take $\mu=\{12\}$ as an example. If one were to define outside options for agents 1 and 2 using the solution of the submarket $(G_{-12}, v)$, then necessarily $d_1=0$ and $d_2=1/2$. These outside options cannot justify the payoff division $(x_1, x_2)=(0,1)$ via Nash bargaining, implying that no recursively defined solution selects from the set of stable outcomes. This example highlights how recursion may render the solution empty, whereas our non-recursive notion of credibility remains well defined.

\section{Results}\label{sec:results}

\subsection{General results}

We first show that if buyer $i$ and seller $j$ are matchable by any optimal matching of the market $(G, v)$, then the sum of any pair of their credible outside options, $d_i$ and $d_j$, is no more than $v_{ij}$. This makes sure that the bargaining problem between them is well-defined. 

\begin{lemma}
	Suppose $ij \in \mu$ for some optimal matching $\mu$ of the market $(G, v)$.  If $(\mu^{\prime},x^{\prime})$ is a stable outcome of the submarket $(G_{-ij},v)$, then
	\[
	x'_{i} + x'_{j} \leq v_{ij}.
	\]
	
\end{lemma}

\begin{proof}
	Suppose instead $x_{i}^{\prime}+x_{j}^{\prime}>v_{ij}$. Then since $(\mu', x')$ is a stable outcome of the submarket $(G_{-ij}, v)$, it is also a stable outcome of the market $(G, v)$. Given that $\mu$ is an optimal matching of $(G, v)$, $(\mu, x')$ is also a stable outcome of $(G, v)$, which implies $x'_i + x'_j =v_{ij}$. We have a contradiction.
\end{proof}

The following result shows that the credible bargaining solution, like stable outcome, is more of a solution on the distribution of payoffs than a solution on the matching of agents.

\begin{proposition}
	If $\mu^{\prime}$ is an optimal matching and $x$ is the payoff vector of any credible bargaining solution $(\mu, x)$, then $(\mu^{\prime}, x)$ is also a credible bargaining solution.
\end{proposition}

\begin{proof}
	Since $\mu'$ is an optimal matching and $(\mu, x)$ is an stable outcome, we know that $(\mu', x)$ is also a stable outcome. Next, consider any pair of agents $ij$ who are matched by $\mu^{\prime}$. We need to show that the division $(x_i, x_j)$ is justifiable as Nash bargaining solution under credible outside options. We discuss two cases.
	
	\textbf{Case I.} $\mu^{\prime}(i)=\mu(i)$. In this case, $i$ is matched with the same agent under both $\mu$ and $\mu^{\prime}$. Since $(\mu,x)$ is a credible bargaining solution of $(G,v)$, the division $(x_{i},x_{j})$ is justifiable by credible outside options from the market $(G_{-ij},v)$.
	
	\textbf{Case II.} $\mu^{\prime}(i)\neq\mu(i)$. (It is possible that $\mu\left(  i\right)=\emptyset.$) Since $(\mu,x)$ is a stable outcome of $(G,v)$ and the pair $ij$ is unmatched under $\mu,(\mu,x)$ must also be a stable outcome under $(G_{-ij},v)$. In addition, since $(\mu,x)$ is a stable outcome of $(G,v)$ and $\mu^{\prime}$ is an optimal matching of $(G,v)$, $(\mu^{\prime},x)$ is also a stable outcome of $(G,v)$. Since $i$ and $j$ are matched under $\mu^{\prime}$, $x_{i}+x_{j}=v_{ij}$. These two facts imply that for $(\mu^{\prime},x)$ and its matched pair $ij,$ the division $(x_{i},x_{j})$ of $v_{ij}$ is justifiable by the outside options $d_i=x_i$ and $d_j = x_j$, which are $i$ and $j$'s stable payoffs in the stable outcome $( \mu,x)$ of $(G_{-ij},v).$	
\end{proof}

Fix a market $(G, v)$. We call a link $ij \in G$ \textbf{essential} if it is matched by every optimal matching. The following result shows that if $(\mu, x)$ is a stable outcome and $ij \in\mu$ is not essential, then $i$ and $j$'s division $(x_{i}, x_{j})$ of $v_{ij}$ is always justifiable.

\begin{proposition}
	\label{prop_justify_essential}  Suppose $(\mu, x)$ is a stable outcome and $ij \in\mu$. If there exists an optimal matching $\mu^{\prime}$ such that $ij \notin\mu^{\prime}$, then there exist credible outside options $d_{i}$ and $d_{j}$ such that
	\[
	(x_{i}, x_{j})=NBS(v_{ij}; d_{i}, d_{j}).
	\]
	
\end{proposition}

\begin{proof}
	Since $(\mu, x)$ is a stable outcome and $\mu'$ is an optimal matching of $(G, v)$, $(\mu', x)$ is also a stable outcome of $(G, v)$. Also, as $ij \notin \mu'$, $(\mu', x)$ is also a stable outcome of $(G_{-ij}, v)$. Therefore, in the market $(G, v)$, at the matching $\mu$, when $i$ and $j$ bargain to divide $v_{ij}$, there exist credible outside options $d_i =x_i$ and $d_j =x_j$. Since $d_i + d_j =v_{ij}$, $NBS(v_{ij}; d_i, d_j)=(x_i, x_j)$.
\end{proof}

Proposition \ref{prop_justify_essential} is an important result. It tells us that to verify whether a stable outcome $(\mu, x)$ is a credible bargaining solution, we only need to verify the justifiability of divisions at the essential links matched by $\mu$.

\begin{corollary} \label{corollary_essential}
	A stable outcome $(\mu, x)$ is a credible bargaining solution if for all essential links $ij \in \mu$, the division$(x_i, x_j)$ is justifiable as the Nash bargaining solution under some credible outside options.
\end{corollary}

\subsection{Existence}
The result below shows that there is always a credible bargaining solution in which every agent receives the average of his/her maximal and minimal stable payoffs in $(G, v)$. 

For any $k\in I\cup J,$ let $\bar{x}_{k}$ and $\underline{x}_{k}$ be agent $k$'s maximal and minimal stable payoffs in $(G,v)$, respectively. 

\begin{theorem} \label{thm_existence}
	Let $\mu$ be any optimal matching of $(G,v)$ and let $x$ be the payoff vector in which
	\[x_k=\frac{\bar{x}_{k}+\underline{x}_k}{2}, \forall k \in I \cup J.\]
	
	Then $(\mu,x)$ is a credible bargaining solution.
\end{theorem}

For $ij \in \mu$, when $i$ and $j$ bargain to divide $v_{ij}$, $i$ wishes to set $d_i=\bar{y}_i$ and $j$ wishes to set $d_j=\bar{y}_j$, where $\bar{y}_i$ and $\bar{y}_j$ are their respective maximal stable payoffs in the submarket $(G_{-ij}, v)$. However, such outside options are likely not credible because they may not be jointly feasible in any stable outcome of $(G_{-ij}, v)$.\footnote{As an example, suppose $ij$ is a link in a square with unit surplus.}

Nonetheless, in Lemma \ref{lemma_compromise}, we show that there exists $\delta\geq0 $ such that $d_{i}=\bar{y}_{i} - \delta$ and $d_{j} =\bar{y}_{j} - \delta$ are credible outside options. We call such $d_i$ and $d_j$ outside options with symmetric compromise. The key idea of the proof of Theorem \ref{thm_existence} is that if $i$ and $j$ divide $v_{ij}$ according to the Nash bargaining solution, then under outside options with symmetric compromise, they each receives the average of their maximal and minimal stable payoffs in $(G, v)$. 

We call the solution presented in Theorem \ref{thm_existence} the \textbf{credible bargaining solution with symmetric compromise}. This solution has appeared in \cite{Thompson81}, \cite{NunezRafels02}, \cite{NunezRafels08}, and \cite{Elliott15}, among others.

\begin{proof} [Proof of Theorem \ref{thm_existence}]
	Suppose $\mu$ is an optimal matching of $(G, v)$. First, note that since $x$ is the average of the payoff vectors of the buyer-optimal stable outcome and the seller-optimal stable outcome, $(\mu, x)$ is also a stable outcome. Next, for every $ij \in\mu$, we construct credible outside options $d_i$ and $d_j$ such that if $i$ and $j$ divide $v_{ij}$ according to the Nash bargaining solution, they will each receive the average of their maximal and minimal stable payoffs in $(G, v)$.
	
	For every $k \in I \cup J$, let $\bar{y}_{k}$ and $\underline{y}_k$ be $k$'s maximal stable payoff and minimal stable payoffs in the submarket $(G_{-ij}, v)$, respectively. When buyer $i$ and seller $j$ bargain to divide $v_{ij}$, $i$ wishes to set $d_{i}=\bar{y}_{i}$ and $j$ wishes to set $d_{j}=\bar{y}_{j}$. However, such outside options are likely not credible because they may not be jointly feasible in any stable outcome of $(G_{-ij}, v)$.
	
	\begin{lemma} \label{lemma_compromise}
		There exists $\delta >0$ such that $d_i=\bar{y}_i - \delta$ and $d_j=\bar{y}_j -\delta$ are credible outside options.
	\end{lemma}
	
	\begin{proof}
		If either $\bar{y}_{i}=\underline{y}_{i}$ or $\bar{y}_{j}=\underline{y}_{j},$ then $\bar{y}_i$ and $\bar{y}_j$ are credible outside options. Therefore, we can simply set $\delta=0$. Otherwise, set the weight $\alpha\in\lbrack0,1]$ to satisfy
		\[
		\frac{1-\alpha}{\alpha}=\frac{\bar{y}_{j}-\underline{y}_{j}}{\bar{y}_{i}-\underline{y}_{i}},
		\]
		
		and set
		\[
		\delta=(1-\alpha)(\bar{y}_{i}-\underline{y}_{i})=\alpha(\bar{y}_{j}-\underline{y}_{j}).
		\]
		
		Then
		\[
		\left(
		\begin{array}
			[c]{c}%
			\bar{y}_{i}-\delta\\
			\bar{y}_{j}-\delta
		\end{array}
		\right)  =\alpha\left(
		\begin{array}
			[c]{c}%
			\bar{y}_{i}\\
			\underline{y}_{j}
		\end{array}
		\right)  +\left(  1-\alpha\right)  \left(
		\begin{array}
			[c]{c}
			\underline{y}_{i}\\
			\bar{y}_{j}%
		\end{array}
		\right)
		\]
		
		Since the payoff vectors of the stable outcomes of $(G_{-ij}, v)$ consist a convex polytope, we know that $\bar{y}_{i}-\delta$ and $\bar{y}_{j}-\delta$ are $i$ and $j$'s stable payoffs at some stable outcome in $(G_{-ij}, v)$, respectively.
	\end{proof}

	Note that at $(d_{i}, d_{j})=(\bar{y}_{i} - \delta, \bar{y}_{j} - \delta)$, since $d_{i}-d_{j}=\bar{y}_{i} - \bar{y}_{j}$, the difference between $i$ and $j$'s outside options is independent of the choice of $\delta$.
	
	With these outside options, when $i$ and $j$ bargaing to divide $v_{ij},$
	their Nash bargaining solution payoffs are, respectively,%
	\begin{align}
		x_{i}  &  =\bar{y}_{i}-\delta+\frac{v_{ij}-(\bar{y}_{i}-\delta)-(\bar{y}_{j}-\delta)}{2}
		=\bar{y}_i + \frac{1}{2}(v_{ij}-\bar{y}_{i}-\bar{y}_{j}), \label{eqn_division_x}\\
		x_{j}  &  =\bar{y}_{j}-\delta+\frac{v_{ij}-(\bar{y}_{i}-\delta)-(\bar{y}_{j}-\delta)}{2} 
		=\bar{y}_j + \frac{1}{2}(v_{ij}-\bar{y}_{j}-\bar{y}_{i}).  \label{eqn_division_y}
	\end{align}

	Note that $\bar{x}_{i},\bar{x}_{j}$ are $i,j$'s marginal contributions in $(G, v)$, respectively, and $\bar{y}_i, \bar{y}_j$ are $i,j$'s marginal contributions in $(G_{-ij}, v)$, respectively.\footnote{See, e.g., Lemma 8.15 of \cite{RothSoto90}.} Therefore,
	\begin{align*}
		\bar{x}_{i}  &  =v_{G}-v_{G_{-i}},\bar{x}_{j}=v_{G}-v_{G_{-j}}, \\  
		\bar{y}_{i}  &  =v_{G_{-ij}}-v_{G_{-i}},\bar{y}_{j}=v_{G_{-ij}}-v_{G_{-j}}.
	\end{align*}

	As a result,
	\[
	\bar{y}_{i}-\bar{y}_{j}=\bar{x}_{i}-\bar{x}_{j}=v_{G_{-j}}-v_{G_{-i}}.
	\]
	
	Together with Equations \ref{eqn_division_x} and \ref{eqn_division_y}, we know that at the Nash bargaining solution $(x_i, x_j)$ under outside options $(d_i, d_j)=(\bar{y}_i -\delta, \bar{y_j}- \delta)$,
	\[x_i - x_j =\bar{y}_i - \bar{y}_j=\bar{x}_i - \bar{x}_j.\]
	
	Since $i$ and $j$ are dividing the surplus $v_{ij}$, $x_i +x_j =v_{ij}$. Therefore, 
	\begin{align}
		x_{i}  &  =\bar{x}_i + \frac{1}{2}(v_{ij}-\bar{x}_{i}-\bar{x}_{j}), \label{eqn_sol_x} \\  
		x_{j}  &  =\bar{x}_j + \frac{1}{2}(v_{ij}-\bar{x}_{j}-\bar{x}_{i}). \label{eqn_sol_y}
	\end{align}
	
	It is well known that at both the buyer-optimal stable outcome and the seller-optimal stable outcome of $(G, v)$, $i$ and $j$'s payoffs divide $v_{ij}$. That is, 
	\[\bar{x}_{i}+\underline{x}_{j}=v_{ij}=\bar{x}_{j}+\underline{x}_{i}.\]
	
	Plug into Equations \ref{eqn_sol_x} and \ref{eqn_sol_y}, we obtain that the Nash bargaining solution between $i$ and $j$ under outside options $(d_{i}, d_{j})=(\bar{y}_{i} - \delta, \bar{y}_{j} - \delta)$ is
	\begin{align*}
		x_{i}  &  =\frac{1}{2}(\bar{x}_{i}+\underline{x}_{i}),\\
		x_{j}  &  =\frac{1}{2}(\bar{x}_{j}+\underline{x}_{j}).
	\end{align*}
	
	We have therefore completed this proof.
\end{proof}

\subsection{Unit-surplus case}

We say that the market $(G, v)$ is a market with unit-surplus if for all $ij \in G, v_{ij}=1$. Recall that a link $ij \in G$ is called \textbf{essential} if it is matched by every optimal matching.

\begin{theorem}\label{thm_characterization}
	Suppose $v_{ij}=1, \forall ij \in G$. An outcome $(\mu, x)$ is a credible bargaining solution of the market $(G, v)$ if and only if it is stable and for all $ij \in \mu$ that is essential, $x_i=x_j=1/2$.
\end{theorem}

The proof of Theorem \ref{thm_characterization} crucially relies on the fact that under unit-surplus, if $ij \in \mu$ is an essential link of $(G, v)$, then there is an optimal matching in $(G_{-ij}, v)$ which leaves both $i$ and $j$ unmatched. This fact implies that in the submarket $(G_{-ij}, v)$, both $i$ and $j$ can only receive zero at all stable outcomes.

\begin{proof}[Proof of Theorem \ref{thm_characterization}]
	
	Let $(\mu, x)$ be any stable outcome of the market $(G, v)$. To prove this theorem, it is sufficient to show that $(\mu, x)$ is a credible bargaining solution if and only if for all $ij \in \mu$ that is essential, $x_i=x_j=1/2$.
	
	Due to Corollary \ref{corollary_essential}, $(\mu, x)$ is a credible bargaining solution if and only if for all essential links $ij \in \mu$, the division $(x_i, x_j)$ is justifiable as the Nash bargaining solution under some credible outside options.
	
	Let $ij \in \mu$ be any essential link of $\mu$. Then the total surplus that can be generated from the submarket $(G_{-ij}, v)$ must be strictly less than that of $(G, v)$. As a result, at the submarket $(G_{-ij}, v)$, the matching $\mu -\{ij\}$ must be optimal, because it generates exactly one less unit of surplus than the total surplus of $(G, v)$. That is, at the market $(G_{-ij}, v)$, there exists an optimal matching that leaves both $i$ and $j$ unmatched. Therefore, in any stable outcome $(\mu', x')$ of $(G_{-ij}, v)$, $x'_i = x'_j =0$.	
	
	As a result, when $i$ and $j$ bargain to divide $v_{ij}$, only $(d_i, d_j)=(0, 0)$ can be a pair of credible outside options, implying that $(x_i, x_j)$ is a justifiable division if and only if $x_i=x_j=1/2$. 
\end{proof}

How to identify the essential links of $(G, v)$? 

The set of nodes in $G$ (i.e., the set of all agents) can be partitioned into under-demanded agents, over-demanded agents, and perfectly matched agents as follows:

\begin{enumerate}
	\item $U$: the set of under-demanded agents in $G$, which consists of agents who are unmatched in at least one optimal matching;
	
	\item $O$: the set of over-demanded agents in $G$, which consists of agent in $(I \cup J) -U$ who are linked to at least one agent in $U$;
	
	\item $P$: the set of perfectly matched agents in $G$, which consists of agents not in $U$ or $O$.\footnote{Our use of terminologies—under-demanded, over-demanded, and perfectly matched agents—follows \cite{Manea16}. Also, the sets $U, O, P$ correspond to the sets $D, A, C$, respectively, in \cite{LovaszPlummer86}.}
\end{enumerate}

Due to the Edmonds-Gallai decomposition theorem (\citealp{LovaszPlummer86}, Theorem 3.2.1), we know that $U$ is an independent set, i.e., no two agents in $U$ are linked. We also know that if $\mu$ is an optimal matching, then it matches every agent in $O$ to a distinct agent in $U$ and it matches agents in $P$ with each other perfectly; moreover, every agent in $U$ is unmatched by some optimal matching. These facts tell us that at any stable outcome $(\mu, x)$, $x_i =0$ for all $i \in U$ and $x_i=1$ for all $i \in O$.

If $ij$ is an essential link of $(G, v)$, it must be that both $i$ and $j$ are in the set $P$. Let's restrict attention to the subgraph consisting of only agents in $P$ and only links among them; denote this subgraph by $G_P$. Call a link between a pair of agents in $P$ forbidden if it is not matched by any optimal matching of $(G, v)$; a link that is not forbidden is called allowed. A network (graph) is called elementary if its allowed links form a connected subgraph of it. 

\begin{lemma}[\citealp{LovaszPlummer86}, Corollary 4.2.10] \label{lemma_Little}
	Every pair of links of an elementary bipartite graph are contained in a (nice) cycle.
\end{lemma} 

Given Lemma \ref{lemma_Little}, it is straightforward to know which links are essential. 

\begin{proposition} \label{prop_essential}
	A link $ij \in G$ is an essential link of $(G, v)$ if and only if it is the only link in an elementary subgraph of $G_P$.
\end{proposition}

\begin{proof}
	If $ij \in \mu$ is in an elementary subgraph of $G_P$ that contains more than one links, then due to Lemma \ref{lemma_Little}, it is part of a cycle and hence there exists another optimal matching that does not contain $ij$. Therefore, $ij$ is not an essential link of $(G, v)$. If $ij \in \mu$ is the only link in an elementary subgraph of $G_P$, then it has to be matched by every optimal matching of $(G, v)$. Therefore, it is an essential link. 
\end{proof}

We use the following example to illustrate the procedure that finds essential links.

\begin{example} \label{example_EG}
	Consider the following market with unit-surplus (See Figure \ref{fig:EG}). There are thirteen agents and they are labeled by the respective set (among $U, O, P$) they belong to in the Edmonds-Gallai decomposition.

	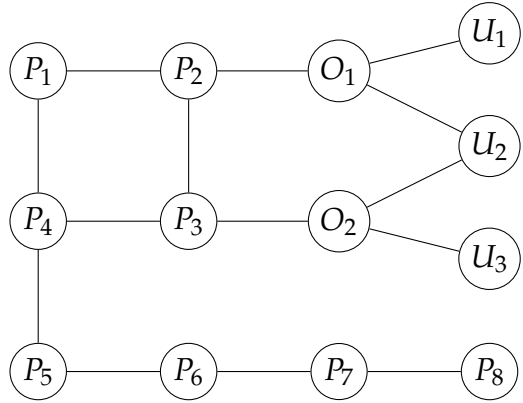
\begin{figure}[htbp]
		\centering
		\begin{tikzpicture}[
			node/.style={circle, draw, inner sep=2pt, minimum size=6mm},
			enclosure/.style={rounded corners, draw, dotted, thick}
			]	
			\begin{scope}[rotate=90]

				% Top row
				\node[node] (5) at (0,0)  {$P_5$};
				\node[node] (4)  at (2,0)  {$P_4$};
				\node[node] (1)  at (4,0)  {$P_1$};
				\node[node] (6) at (0,-2) {$P_6$};
				\node[node] (3) at (2,-2) {$P_3$};
				\node[node] (2) at (4,-2) {$P_2$};
				\node[node] (7) at (0,-4) {$P_7$};
				
				\node[node] (10) at (2,-4) {$O_2$};
				\node[node] (9) at (4,-4) {$O_1$};
				
				\node[node] (8) at (0,-6) {$P_8$};
				\node[node] (13) at (1.5,-6) {$U_3$};
				\node[node] (12) at (3,-6) {$U_2$};
				\node[node] (11) at (4.5,-6) {$U_1$};
				
				% Edges top to second layer
				\draw (5) -- (4);
				\draw (4) -- (1);
				
				\draw (5) -- (6);
				\draw (4)  -- (3);
				\draw (1)  -- (2);
				
				% Middle square edges
				
				\draw (3) -- (2);
				
				% Downward chains
				\draw (6) -- (7);
				\draw (7) -- (8);
				
				\draw (3) -- (10);
				\draw (2) -- (9);
				
				\draw (10) -- (13);
				
				% Branches on right side
				\draw (9) -- (11);
				\draw (12) -- (10);
				\draw (9) -- (12);

			\end{scope}
			
		\end{tikzpicture}
		\caption{Edmonds-Gallai decomposition}
		\label{fig:EG}
	\end{figure}
	
Restricting attention to $G_P$ (the subgraph which consists of only agents in $P$), there are two forbidden links, namely $P_4P_5$ and $P_6P_7$. After removing them, we obtain three elementary subgraphs (see Figure \ref{fig:Elementary}).

Each of $P_5P_6$ and $P_7P_8$ is the only link in an elementary subgraph. Therefore, due to Proposition \ref{prop_essential}, both of them are essential links. As a result, from Theorem \ref{thm_characterization}, we know that an outcome $(\mu, x)$ is a credible bargaining solution of this market if and only if it is stable and agents $P_5, P_6, P_7$ and $P_8$ all receive $1/2$.
	
		\begin{figure}[htbp]
		\centering
		\begin{tikzpicture}[
			node/.style={circle, draw, inner sep=2pt, minimum size=6mm},
			enclosure/.style={rounded corners, draw, dotted, thick}
			]	
			\begin{scope}[rotate=90]

				\node[node] (5) at (0,0)  {$P_5$};
				\node[node] (4)  at (2,0)  {$P_4$};
				\node[node] (1)  at (4,0)  {$P_1$};
				\node[node] (6) at (0,-2) {$P_6$};
				\node[node] (3) at (2,-2) {$P_3$};
				\node[node] (2) at (4,-2) {$P_2$};
				\node[node] (7) at (0,-4) {$P_7$};
                \node[node] (8) at (0,-6) {$P_8$};
				
				\draw (4) -- (1);
				
				\draw (5) -- (6);
				\draw (4)  -- (3);
				\draw (1)  -- (2);
				
				% Middle square edges
				
				\draw (3) -- (2);
				\draw (7) -- (8);

			\end{scope}
			
		\end{tikzpicture}
		\caption{Elementary subgraphs of $G_P$}
		\label{fig:Elementary}
	\end{figure}
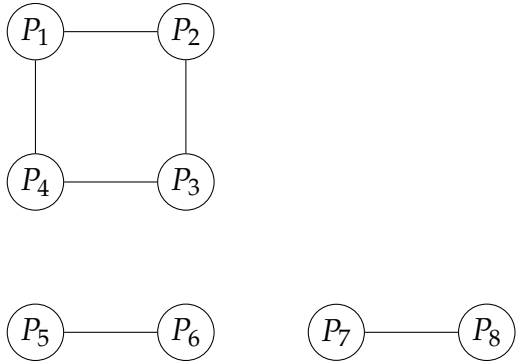

\end{example}

\begin{remark}
	Fix an Edmonds–Gallai decomposition of $G$. Most cooperative solution concepts and noncooperative bargaining models predict that all agents in $U$ receive payoff $0$ and all agents in $O$ receive payoff $1$, while their predictions may differ for agents in $P$. For these perfectly matched agents, beyond the requirements imposed by stability, the credible bargaining solution only pins down the payoffs of agents who are connected by essential links, requiring each such agent to receive $1/2$.
	
	In comparison, under the dynamic bargaining process proposed by \cite{Corominas04}, there exists a subgame-perfect equilibrium in which all agents in $P$ receive $1/2$ in the limit as the discount factor $\delta \rightarrow 1$. \cite{Polanski07} obtains similar predictions in his equilibrium analysis of the dynamic bargaining process he proposes. By contrast, under decentralized bargaining protocols, \cite{AbreuManea12GEB} show that Markov perfect equilibria may yield highly asymmetric payoffs even among perfectly matched agents (see Figure 3 in \cite{AbreuManea12GEB} for example and \cite{Manea16} for a comprehensive review on noncooperative bargaining).
	  
\end{remark}

\section{Conclusion}\label{sec:concl}

This paper studies surplus division in network-constrained matching markets by introducing a new bargaining solution based on credible outside options. Building on the classical notion of stability, our approach selects among stable outcomes by requiring that surplus within each matched buyer–seller pair be justifiable as the outcome of Nash bargaining with outside options that arise from stable outcomes of appropriately defined submarkets. This additional credibility requirement disciplines surplus division without imposing an explicit noncooperative bargaining protocol and yields sharp implications for equilibrium payoffs, particularly on essential links.

We establish general properties of the credible bargaining solution, including existence, and provide a complete characterization in the unit-surplus case. Our analysis highlights the central role of essential links in determining bargaining outcomes and shows that, away from these links, stability alone suffices to justify surplus division. 

Several directions for future research remain open. On the theoretical side, it would be of interest to further explore the structural properties of the credible bargaining solution and to extend the characterization beyond the unit-surplus environment. More broadly, the notion of credibility based on submarket stability suggests a promising avenue for extending our approach to more general cooperative games, in which outside options are endogenously determined by the outcomes of reduced games. Finally, it would be worthwhile to investigate normative criteria that could further refine the credible bargaining solution or guide its selection among multiple stable outcomes.


\bigskip




\begin{thebibliography}{999999999999999999999999999999999999999999999999999999999999999999999999999999999999999999999999} %


\bibitem[Abreu and Manea(2012a)]{AbreuManea12GEB}{\small D. Abreu and M. Manea (2012a), Markov equilibria in a model of bargaining in networks, \emph{Games and Economic Behavior}, 75, 1-16.}

\bibitem[Abreu and Manea(2012b)]{AbreuManea12JET}{\small D. Abreu and M. Manea (2012b), Bargaining and efficiency in networks, \emph{Journal of Economic Theory}, 147, 43-70.}

\bibitem[Aristotle(2004)]{Aristotle04} {\small Aristotle (2004), Nicomachean ethics, Translated and edited by Roger Crisp, \emph{Cambridge University Press}.}

\bibitem[Aumann and Maschler(1985)]{AumannMaschler85}{\small R. Aumann and M. Maschler (1985), Game-theoretic analysis of a bankruptcy problem from the Talmud, \emph{Journal of Economic Theory}, 36, 195–213.}

\bibitem[Cook and Yamagishi(1992)]{CookYamagishi92}{\small K. S. Cook and T. Yamagishi (1992), Power in exchange networks: A power-dependence formulation, \emph{Social Networks}, 14, 245--265.}

\bibitem[Corominas-Bosch(2004)]{Corominas04}{\small M. Corominas-Bosch (2004), Bargaining in a network of buyers and sellers, \emph{Journal of Economic Theory}, 115, 35-77.}

\bibitem[Crawford and Rochford(1986)]{Crawford-Rochford86}{\small V. P. Crawford and S. C. Rochford (1986), Bargaining and competition in matching markets, \emph{International Economic Review}, 27, 329–348.}

\bibitem[Davis and Maschler(1965)]{Davis-Maschler65}{\small M. Davis and M. Maschler (1965), The kernel of a cooperative game, \emph{Naval Research Logistics Quarterly}, 12, 223-259.}

\bibitem[Driessen(1998)]{Driessen98}{\small T. S. H. Driessen (1998), A note on the inclusion of the kernel in the core of the bilateral assignment game, \emph{Internaltional Journal of Game Theory}, 27, 301-303.}

\bibitem[Elliott(2015)]{Elliott15}{\small M. Elliott (2015), Inefficiencies in networked markets, \emph{American Economic Journal: Microeconomics}, 7, 43-82.}

\bibitem[Harsanyi(1956)]{Harsanyi56}{\small J. Harsanyi (1956), Approaches to the bargaining broblem before and after the theory of games: A critical discussion on Zeuthen's, Hick's, and Nash's theories, \emph{Econometrica}, 24, 144-157.}

\bibitem[Harsanyi(1977)]{Harsanyi77}{\small J. Harsanyi (1977), Rational behavior and bargaining equilibrium in games and social situations, \emph{Cambridge University Press}.}

\bibitem[Kleinberg and Tardos(2008)]{KleinbergTardos08}{\small J. Kleinberg and E. Tardos (2008), Balanced outcomes in social exchange networks, \emph{Proc. 40th ACM Symposium on Theory of Computing}.}

\bibitem[Lov\'{a}sz and Plummer(1986)]{LovaszPlummer86}{\small L. Lov\'{a}sz and M. D. Plummer (1986), Matching theory, \emph{North-Holland}.}

\bibitem[Manea(2016)]{Manea16}{\small M. Manea (2016), Models of bilateral trade in networks, }in {\small \emph{The Oxford Handbook on the Economics of Networks}, edited by Y. Bramoull\'{e}, A. Galeotti and B. W. Rogers, \emph{Oxford: Oxford University Press}.}

\bibitem[Nash(1950)]{Nash50}{\small J. Nash (1950), The bargaining problem, \emph{Econometrica}, 18, 155-162.}

\bibitem[Nash(1953)]{Nash53}{\small J. Nash (1953), Two-person cooperative games, \emph{Econometrica}, 21, 128-140.}

\bibitem[N{\small\'{u}\~{n}}ez and Rafels(2002)]{NunezRafels02}{\small M. N\'{u}\~{n}ez and C. Rafels (2002), The assignment game: the }$\tau$-value{\small, \emph{International Journal of Game Theory}, 31, 411-422.}

\bibitem[N{\small\'{u}\~{n}}ez and Rafels(2008)]{NunezRafels08}{\small M. N\'{u}\~{n}ez and C. Rafels (2008), A cooperative bargaining approach to the assignment market, \emph{Group Decision and Negotiation}, 17, 553-563.}

\bibitem[N{\small \'{u}\~{n}}ez and Rafels(2015)]{NunezRafels15}{\small M. N\'{u}\~{n}ez and C. Rafels (2015), A survey on assignment markets, \emph{Journal of Dynamics and Games}, 2, 227-256.}

\bibitem[Polanski(2007)]{Polanski07}{\small A. Polanski (2007), Bilateral bargaining in networks, \emph{Journal of Economic Theory}, 134, 557-565.}

\bibitem[Rochford(1984)]{Rochford84}{\small S. C. Rochford (1984), Symmetrically pairwise-bargained allocations in an assignment market, \emph{Journal of Economic Theory}, 34, 262-281.}

\bibitem[Roth and Sotomayor(1988)]{RothSoto88}{\small A. Roth and M. Sotomayor(1988), Interior points in the core of two-sided matching markets, \emph{Journal of Economic Theory}, 45, 85-101.}

\bibitem[Roth and Sotomayor(1990)]{RothSoto90}{\small A. Roth and M. Sotomayor (1990), Two-sided matching: A study in game-theoretic modeling and analysis, \emph{Cambridge University Press}.}

\bibitem[Shapley and Shubik(1972)]{ShapleyShubik72}{\small L. Shapley and M. Shubik (1972), The assignment game I: the core, \emph{International Journal of Game Theory}, 1, 111-130.}

\bibitem[Thompson(1981)]{Thompson81}{\small G. L. Thompson (1981), Auctions and market games, In: R. Aumann (ed), \emph{Essays in Game Theory in honor of Oskar Morgenstern}, Mannheim, Germany, 181--196.}

\end{thebibliography}
\end{document}